\documentclass[submission]{eptcs}
\usepackage{amssymb}
\usepackage{amsmath}
\usepackage{amsfonts}
\usepackage{amsthm}
\usepackage{cleveref}
\usepackage{thmtools, thm-restate}

\usepackage[draft,footnote,nomargin]{fixme}

\newcommand{\N}{\mathbb{N}}
\newcommand{\Q}{\mathbb{Q}}
\newcommand{\R}{\mathbb{R}}
\newcommand{\F}{\mathbb{F}}
\newcommand{\Fplus}{\mathbb{F}_{\geq 0}}
\newcommand{\Rplus}{\mathbb{R}_{\geq 0}}
\renewcommand{\H}{\mathcal{H}}

\newcommand{\tr}{\mathrm{tr}}
\newcommand{\defin}{:=}
\newcommand{\Pl}{\mathrm{Pl}} 
\newcommand{\Po}{\mathrm{Po}} 

\declaretheorem[numberwithin=section, name=Theorem]{thm}

\declaretheorem[sibling=thm, name=Corollary]{cor}
\declaretheorem[sibling=thm, name=Definition, style=definition]{defn}
\declaretheorem[sibling=thm, name=Example, style=remark]{expl}

\allowdisplaybreaks

\begin{document}

\title{Plausibility measures on test spaces} \author{Tobias Fritz
  \qquad\qquad\qquad Matthew S. Leifer \institute{Perimeter Institute for
    Theoretical Physics, Waterloo ON, Canada\thanks{Research at
      Perimeter Institute is supported in part by the Government of
      Canada through NSERC and by the Province of Ontario through
      MRI\@. TF is supported by the John Templeton Foundation. ML is supported by the Foundational Questions Institute
      (FQXi).}}  \email{tfritz@perimeterinstitute.ca
    \quad matt@mattleifer.info} }

\newcommand{\titlerunning}{Plausibility measures on test spaces}

\newcommand{\authorrunning}{T. Fritz and M. S. Leifer}

\newcommand{\event}{QPL 2015}

\maketitle

\begin{abstract}
  Plausibility measures are structures for reasoning in the face of
  uncertainty that generalize probabilities, unifying them with weaker
  structures like possibility measures and comparative probability
  relations.  So far, the theory of plausibility measures has only
  been developed for classical sample spaces.  In this paper, we
  generalize the theory to test spaces, so that they can be applied to
  general operational theories, and to quantum theory in particular.
  Our main results are two theorems on when a plausibility measure
  \emph{agrees} with a probability measure, i.e.\ when its comparative
  relations coincide with those of a probability measure.  For
  strictly finite test spaces we obtain a precise analogue of the
  classical result that the \emph{Archimedean} condition is necessary
  and sufficient for agreement between a plausibility and a
  probability measure.  In the locally finite case, we prove a
  slightly weaker result that the Archimedean condition implies almost
  agreement.
\end{abstract}

\section{Introduction}

It is often stated that a physical theory must (at least) supply
probabilities for the outcomes of experiments, but are probabilities
strictly necessary, or can we sometimes get away with a weaker
predictive structure?  Recent years have seen a growth of interest in
theories that are variously called \emph{possibilistic}
\cite{Fritz2009, Mansfield2012}, \emph{modal} \cite{Schumacher2010,
  Schumacher2010a, Schumacher2012}, or \emph{relational}
\cite{Abramsky2013}, in which the continuum of probability assignments
is replaced by a two valued assessment of whether or not an outcome is
possible.

If we are serious about using such weaker structures in fundamental
physics, then we need to employ a theory that is capable of
reproducing probabilistic predictions where they are known to work
well, i.e.\ in most ordinary experiments in quantum theory and
statistical mechanics, but allows for weaker predictions in general.
In the classical case, the theory of \emph{plausibility measures}
\cite{Friedman1995, Friedman2001, Halpern2001, Halpern2001a} can play
this role.  Both possibility and probability measures are examples of
plausibility measures, but in general plausibility measures allow for
comparative assessments, i.e.\ we may be able to say that $A$ is more
likely than $B$, without specifying a precise numerical probability
for either.  Possibility measures are obtained when there are only two
plausibility values and, as we know from the foundations of subjective
probability theory \cite{Fishburn1986}, a probability measure can be
derived when the ordering relation over events is rich enough.
Plausibility measures are therefore a good framework for unifying and
generalizing the various predictive structures that could be applied
to physical theories, but, so far, they have been limited to classical
theories.

In this paper, we generalize the theory of plausibility measures to
\emph{test spaces}, which are capable of representing nonclassical
theories such as quantum theory.  The fundamental question we address
is when a plausibility measure \emph{agrees} with a probabilistic
state on the test space, i.e.\ when the comparative relations are rich
enough to be faithfully represented by probabilities.  For finite test
spaces, we prove a precise analogue of the classical result that the
\emph{Archimedean} condition is necessary and sufficient for agreement
\cite{Kraft1959, Scott1964, Fishburn1986}.  For locally finite test
spaces, we prove that the Archimedean condition implies a weaker notion
called \emph{almost} agreement.  This is an important first step in
the program of generalizing physics beyond probabilities, as it
enables us to identify the limits in which probability theory applies.

Before embarking on the study of more general predictive structures,
it is worth mentioning at least one of the motivations.  It has been
argued that argued that, along with spacetime, probabilities will
inevitably become fuzzy in future theories of quantum gravity
\cite{Buniy2005, Buniy2006, Mueller2009}.  One proposal for
implementing this is to discretize the state space of quantum systems
so that, for example, a system with a finite dimensional Hilbert
space, such as a spin-$1/2$ particle, would only be able to occupy a
finite number of states \cite{Buniy2005}.  This violates the
superposition principle and raises problems such as how the system
decides which state to ``snap to'' when an observable is measured for
which the eigenvectors are not allowed states.

However, from a test space perspective, the state space of a theory is
derived from the attempt to consistently assign probability measures
to the outcomes of measurements, such that outcomes of measurements
that are physically identified are always assigned the same
probability.  For example, in quantum theory, Gleason's theorem
\cite{Gleason1957} says that any such assignment can be represented by
a density operator on Hilbert space.  If we no longer have the
continuum of probability assignments, then we no longer have the right
to assert that the state space of a quantum system must be a Hilbert
space, let alone a discrete collection of vectors within one.  In our
view, if we are interested in employing fuzzified probabilities in
physics, it is better to start again from scratch with a well-defined
discrete predictive structure, such as a possibility measure or a
comparative plausibility measure, and let the theory tell us what the
state space must be.  This is likely to lead to a more consistent
theory, without the need to introduce arbitrary rules for how systems
are updated upon measurement.  The present work can be viewed as an
attempt to take the first few steps along this road.

The remainder of this paper is structured as follows.  \S\ref{Plaus}
introduces the definitions of test spaces and plausibility measures,
and gives examples.  \S\ref{Agree} defines various notions of
agreement between plausibility measures and probability measures and
states the main results.  \S\ref{Hahn} develops a Hahn-Banach theorem for order unit spaces, which we then use in
\S\ref{Proofs} to prove the main results.  \S\ref{Related} discusses
related work and \S\ref{Conc} concludes with a discussion of possible
future directions.

\section{Plausibility measures on test spaces}

\label{Plaus}

Intuitively speaking, a test space is a generalization of a measurable
space, designed to allow for incompatible events that cannot be
resolved in a single experiment.  It has its origins in the work of
Foulis and Randall \cite{Foulis1972, Randall1973} (see
\cite{Wilce2009} for a review).

\begin{defn}
  A \emph{test space} $(X,\Sigma)$ consists of a set $X$ together with
  a set of subsets $\Sigma\subseteq 2^X$ such that the members of
  $\Sigma$ cover $X$, i.e.~$\bigcup_{T\in\Sigma} T = X$.
\end{defn}

The elements of $X$ are called \emph{outcomes}, while the elements of
$\Sigma$ are the \emph{tests}. Each test represents the set of
possible outcomes of a measurement that can be performed on the
system.  A subset $A\subseteq X$ is called an \emph{event} if it is a
subset of some test. The assumption that $\Sigma$ covers $X$ can be
rephrased as stating that every singleton set $\{x\}$ is an event. We
write $E(X,\Sigma)$ for the set of all events. Under subset inclusion,
$E(X,\Sigma)$ is a partially ordered set with $\emptyset$ as the least
element and all tests as maximal elements.

\begin{expl}
  For a classical system with a finite sample space $X$, the test
  space is $(X,\{X\})$, i.e.\ there is only one test that contains all
  possible outcomes.  $X$ consists of all the possible configurations
  of the system and its elements can, in principle, be resolved by a
  single measurement.  The events are $E(X,\{X\}) = 2^X$, as they
  would be in classical probability theory with a finite sample space.
\end{expl}
\begin{expl}
  For a quantum system with finite dimensional Hilbert space $\H$, the
  test space is $(P(\H),\allowbreak b(\H))$, where $P(\H)$ is the
  projective Hilbert space of $\H$, or equivalently the set of unit
  vectors in $\H$ with vectors differing by a global phase identified,
  and $b(\H)$ is the set of rank-1 projector valued measures, or
  equivalently the set of orthonormal bases of $\H$ with basis vectors
  identified if they differ by a global phase.  Each test represents a
  measurement on the system that is as fine-grained as possible.
  Events correspond to projection operators that may be of higher
  rank.
\end{expl}

\begin{defn}
  A test space $(X,\Sigma)$ is \emph{locally finite} if every test
  $T\in\Sigma$ is finite.
\end{defn}

In what follows, all test spaces are assumed locally finite.  This
includes finite classical systems and quantum systems with finite
dimensional Hilbert spaces.  Note that the outcome set $X$ can be
infinite in a locally finite test space if there are an infinite
number of tests, as there are in the quantum case.

\begin{defn}
  A \emph{probability measure} on a test space $(X,\Sigma)$ is a
  function $\mu:X\to\Rplus$ such that $\sum_{x\in T} \mu(x)=1$ for
  every test $T\in\Sigma$.

  A probability measure can be extended to a function $E(X,\Sigma)
  \to \Rplus$, which we also denote $\mu$, via $\mu(A) = \sum_{x \in
  A} \mu(x)$.
\end{defn}

Note that the same outcome may appear in more than one test in a test
space, e.g.\ $(\{x,y,z\},\{\{x,y\},\allowbreak \{y,z\},\{z,x\}\})$.
This means that, for whatever reason, the outcome $x$ is thought to
have the same physical meaning regardless of whether it appears in
test $\{x,y\}$ or $\{z,x\}$.  Since $\mu$ is a function of $X$, the
probability assigned to an outcome does not depend on which test is
being measured.  This feature is often called \emph{noncontextuality}
of probability assignments.

\begin{expl}
  A probability measure on a classical test space $(X, \{X\})$ with finite $X$ is just
  a probability measure on $(X,2^X)$ in the sense of classical
  probability theory.
\end{expl}

\begin{expl}
  By Gleason's theorem \cite{Gleason1957}, for a Hilbert space $\H$
  with $\mathrm{dim}(\H) \geq 3$, a probability measure on the quantum
  test space $(P(\H),b(\H))$ is of the form $\mu(\Pi) = \tr(\Pi \rho)$
  where $\Pi$ is a projection operator and $\rho$ is a density
  operator (semi-positive operator with $\tr(\rho) = 1$).
\end{expl}

Note that a probability measure is usually called a \emph{state} in
the test space literature since, in the quantum case, it corresponds
to the usual notion of a quantum state.  However, plausibility
measures have equal claim to be thought of as physical states, so we
do not use this terminology here.

\begin{defn}
  A \emph{plausibility measure} on a test space $(X,\Sigma)$ is a
  function $\Pl:E(X,\Sigma)\to D$, where $(D,\preceq)$ is a partially
  ordered set of \emph{plausibility values}, satisfying the following
  conditions:
  \begin{enumerate}
  \item For all $T,R \in \Sigma$, \quad $\Pl(T) = \Pl(R)$. \label{pl1} 
  \item If $A$ and $B$ are events with $A\subseteq B$, then
    $\Pl(A)\preceq \Pl(B)$. \label{pl2}
  \item For any $T \in \Sigma$, \quad $\Pl(\emptyset) \prec
    \Pl(T)$. \label{pl3}
  \end{enumerate}
\end{defn}

Here, we write ``$\Pl(\emptyset) \prec \Pl(T)$'' as shorthand for
``$\Pl(\emptyset) \preceq \Pl(T)$ and $\Pl(T) \not\preceq
\Pl(\emptyset)$''.

This definition is the natural generalization of a classical
plausibility measure \cite{Friedman1995, Friedman2001, Halpern2001} to
test spaces, and we recover the classical case for the test space
$(X,\{X\})$.  Axiom~\ref{pl3} is not always imposed, but it prevents
the trivial case where every event has the same plausibility.

Under these axioms, the image of $\Pl$ is a bounded poset with minimal
element $\Pl(\emptyset)$ and maximal element $\Pl(T)$, where $T \in
\Sigma$ is any test.  Because of this, without loss of generality, we
may assume that $D$ is a bounded poset, with minimal element $0$ and
maximal element $1$, and demand that $\Pl(\emptyset) = 0$ and $\Pl(T)
= 1$ for any test $T \in \Sigma$.  Axiom~\ref{pl3} then becomes $0
\neq 1$.

\begin{expl}
  If $D=\{0,1\}$ with $0\prec 1$, then $\Pl$ is a \emph{possibility measure}.  $0$ is
  interpreted as impossibility and $1$ as possibility.  Note that,
  given any plausibility measure $\Pl$, we may derive a possibility
  measure $\Po$ from it by setting $\Po(A) = 0$ if $\Pl(A) = 0$ and
  $\Po(A) = 1$ otherwise.
\end{expl}

\begin{expl}
  Any probability measure $\mu$ on a test space is also a
  plausibility measure, where the plausibility values are the unit
  interval $[0,1]$, which is totally ordered.
\end{expl}

\begin{expl}
  One way in which plausibility measures can arise is if we have a set
  $\{\mu_i\}$ of possible probability measures on a test space, but no
  prior probability distribution over them.  One can then set $\Pl(A)
  \preceq \Pl(B)$ iff $\mu_i(A) \leq \mu_i(B)$ for every measure in
  the set.
\end{expl}

\section{Agreement}

\label{Agree}

We are interested in determining the conditions under which a
plausibility measure can be faithfully represented by a probability
measure.  In the classical case, this question has been studied
extensively \cite{Fishburn1986}, from which we adapt the following
definitions.

\begin{defn}
  A plausibility measure $\Pl$ on a test space $(X,\Sigma)$
  \emph{agrees} with a probability measure $\mu$ if
  \begin{equation}
    \label{agree}
    \Pl(A) \preceq \Pl(B) \quad \Longleftrightarrow \quad \mu(A) \leq
    \mu(B). 
  \end{equation}
  $\Pl$ \emph{almost agrees} with $\mu$ if
  \begin{equation}
    \Pl(A) \preceq \Pl(B) \quad \Longrightarrow \quad \mu(A) \leq \mu(B)
  \end{equation}
\end{defn}

Note that \cref{agree} is equivalent to
  \begin{equation}
    \Pl(A) \prec \Pl(B) \quad \Longleftrightarrow \quad \mu(A) <
    \mu(B). 
  \end{equation}

Agreement implies that the image of $\Pl$ is totally ordered.  In
contrast, almost agreement is quite weak as, not only does it not
imply total ordering, but it also allows $\mu(A) = 0$ when $\Pl(A)
\succ 0$.

\begin{expl}
  \label{agreefail}
  Finding a $\mu$ that agrees with $\Pl$ is not always possible.  For
  instance, the test space $(\{x,y,z\},\{\{x,y\},\{y,z\},\{z,x\}\})$
  has a possibility measure defined by $\Pl(x) = 0$, $\Pl(y) = 1$ and
  $\Pl(z) = 1$, with the assignments to other events entailed by
  axiom~\ref{pl2}.  However, there is no probability measure with
  $\mu(x) = 0$ and both $\mu(y) > 0$ and $\mu(z) > 0$.  This is
  because we must have $\mu(x) + \mu(y) = 1$ and $\mu(y) + \mu(z) =
  1$, but these assignments would imply $\mu(x) + \mu(y) = 0 + 1 -
  \mu(z) < 1$.
\end{expl}

\begin{expl}
  \label{kps-ex}
  Even in the classical case, agreement is not always possible, as was
  shown by Kraft, Pratt and Seidenberg \cite{Kraft1959}.  For the test
  space $(\{1,2,3,4,5\},\{\{1,2,3,4,5\}\})$, consider the following
  relations between plausibility values
  \begin{align}
    \Pl(\{1,3\}) & \prec \Pl(\{4\}), \\
    \Pl(\{1,4\}) & \prec \Pl(\{2,3\}), \\
    \Pl(\{3,4\}) & \prec \Pl(\{1,5\}), \\
    \Pl(\{2,5\}) & \prec \Pl(\{1,3,4\}).
  \end{align}
  If there exists a probability measure $\mu$ that agrees with these
  assignments then, denoting $\mu(j)$ by $\mu_j$, because of
  additivity we must have
  \begin{align}
    \mu_1 + \mu_3 & < \mu_4, \\
    \mu_1 + \mu_4 & < \mu_2 + \mu_3, \\
    \mu_3 + \mu_4 & < \mu_1 + \mu_5, \\
    \mu_2 + \mu_5 & < \mu_1 + \mu_3 + \mu_4,
  \end{align}
  but summing these inequalities and cancelling terms leads to $0 <
  0$, so no such probability measure exists.
\end{expl}

Clearly then, the existence of an agreeing probability measure
requires additional constraints to be imposed on plausibility
measures.  One such condition is as follows.  Let $(A_1,\ldots,A_n)$
and $(B_1,\ldots,B_n)$ be families of events such that every outcome
in $X$ occurs the same number of times in both, then, for a
probability measure
\[
\mu(A_1) \leq \mu(B_1),\quad\ldots,\quad \mu(A_{n-1}) \leq
\mu(B_{n-1}) \quad\Rightarrow\quad \mu(A_n) \geq \mu(B_n).
\]
Consequently, in order for a plausibility measure to agree with $\mu$,
it must satisfy the analogous condition:
\begin{equation}
  \label{archim}
  \Pl(A_1) \preceq \Pl(B_1),\quad\ldots,\quad \Pl(A_{n-1}) \preceq
  \Pl(B_{n-1}) \quad\Rightarrow\quad \Pl(A_n) \succeq \Pl(B_n).  
\end{equation} 

\begin{defn}
  A plausibility measure on a test space $(X,\Sigma)$ is
  \emph{Archimedean}\footnote{This is sometimes also called
    \emph{strong additivity} in the literature.} if it satisfies
  \cref{archim} for every pair $(A_1,\ldots,A_n)$ and
  $(B_1,\ldots,B_n)$ of families of events such that every outcome in
  $X$ occurs the same number of times in both.
\end{defn}

\begin{expl}
  The plausibility measure in \cref{agreefail} is non-Archimedean.
  Consider the families of events $(\{y,z\},\{x,z\},\{x\})$ and
  $(\{x,y\},\{x,z\},\{z\})$ in which $x$, $y$ and $z$ appear the same
  number of times.  The possibility measure with $\Pl(x) = 0$, $\Pl(y)
  = 1$ and $\Pl(z) = 1$ satisfies
  \begin{align}
    \Pl(\{y,z\}) & \preceq \Pl(\{x,y\}) & \Pl(\{x,z\}) & \preceq
    \Pl(\{x,z\}),
  \end{align}
  but $\Pl(x) \prec \Pl(z)$.
\end{expl}

\begin{expl}  
  For \cref{kps-ex}, the families $(\{1,3\}, \{1,4\}, \{3,4\},
  \{2,5\})$ and $(\{4\}, \{2,3\}, \{1,5\}, \{1,3,4\})$ fail the
  Archimedean condition.
\end{expl}

In our terminology, a classic theorem in comparative probability
\cite{Kraft1959, Scott1964, Fishburn1986} states
\begin{thm}
  A plausibility measure $\Pl$ on a finite classical test space $(X,\{X\})$
  agrees with some probability measure $\mu$ iff the image of $\Pl$ is
  totally ordered and $\Pl$ is Archimedean.
\end{thm}

Our task is to generalize this to locally finite test spaces.  This
works perfectly for test spaces with finite $X$, but so far we have
only been able to prove a weaker result in the locally finite case.

\begin{restatable}{thm}{mainloc}
  \label{mainloc}%
  A plausibility measure $\Pl$ on a locally finite test space
  $(X,\Sigma)$ almost agrees with some probability measure $\mu$ if
  $\Pl$ is Archimedean.
\end{restatable}

\begin{restatable}{thm}{mainfin}
  \label{mainfin}%
  A plausibility measure $\Pl$ on a test space $(X,\Sigma)$ with $X$
  finite agrees with some probability measure $\mu$ iff the image of
  $\Pl$ is totally ordered and $\Pl$ is Archimedean.
\end{restatable}

Unfortunately, we do not know whether the stronger result still holds
for general locally finite test spaces, but it seems likely that
additional topological assumptions may be required.

Our proof strategy is to translate the problem into the language of
order unit spaces over the rational field and make use of the
associated Hahn-Banach theorems.  Before proceeding, we therefore
review some of this theory.

\section{Hahn-Banach theorems for order unit spaces}

\label{Hahn}

Here we redevelop two standard Hahn-Banach theorems for order unit spaces (see e.g.~\cite{PT}) with minor
modifications. Let $V$ be a vector space over an ordered field $\F$. A
subset $C\subseteq V$ is a \emph{convex cone} if it is closed under
addition and positive scalar multiplication,
\begin{align*}
  a\in C,\: b\in C \:\Rightarrow\: a+b\in C && \lambda \in\Fplus,\:
  a\in C \:\Rightarrow\: \lambda a\in C.
\end{align*}
If the convex cone $C$ is clear from the context, then we also write
$a\geq b$ as shorthand for $a-b\in C$. It is straightforward to see
that `$\geq$' defines a partial ordering on $V$ with $a\geq 0$ if and
only if $a\in C$.

\begin{defn}
  An \emph{order unit space} is a triple $(V,C,u)$ such that
  $C\subseteq V$ is a convex cone and $u\geq 0$ is a distinguished
  element called the \emph{order unit} such that
  \begin{enumerate}
  \item $-u\not\geq 0$,
  \item\label{orderunit} For any $a\in V$, there is $\lambda\in\F$ such that
    $\lambda u + a \geq 0$.
  \end{enumerate}
\end{defn}

Axiom~\ref{orderunit} states exactly that every $a\in V$ can be lower bounded by a scalar multiple of $u$.

\begin{defn}
  A \emph{probability measure} $\rho$ on an order unit space $(V,C,u)$
  is a linear functional $\rho:V\to\R$ with $\rho(a)\geq 0$ for $a\geq
  0$ and $\rho(u)=1$.
\end{defn}

Note that the standard terminology for such a functional is ``state'', but we prefer ``probability measure'' for the sake of consistency with our test space terminology.

\begin{thm}[Hahn-Banach extension theorem]
  \label{HBextend}
  Let $(V,C,u)$ be an order unit space. If $U\subseteq V$ is a
  subspace with $u \in U$, then $(U,C\cap U,u)$ is again an order unit
  space. Any probability measure $\sigma$ on $U$ can be extended to a
  probability measure $\rho$ on $V$, i.e.~there is a probability
  measure $\rho:V\to\R$ such that $\rho|_U=\sigma$.
\end{thm}

\begin{proof}
  The claim that $(U,C\cap U,u)$ is again an order unit space is
  straightforward to check; the non-trivial part is the second
  statement. In the following, we present the standard argument
  adapted to our particular situation.

  Consider the collection of all pairs $(W,\rho)$ where $W\subseteq V$
  is a subspace with $u \in W$ and $\rho$ is a probability measure on
  $W$. This set is partially ordered by defining $(W,\rho)\leq
  (W',\rho')$ to mean that $W\subseteq W'$ and $\rho'|_W =
  \rho$. Since the hypothesis of Zorn's lemma trivially holds, it
  follows that this partially ordered set has a maximal
  element. Likewise, if we consider only all the elements ``above''
  the given $(U,\sigma)$, then Zorn's lemma still applies and we
  obtain that there exists a maximal element which is greater than or
  equal to $(U,\sigma)$. Let $(W,\rho)$ be such a maximal
  element. Then we have $\rho|_U=\sigma$ by construction.

  Our goal is to show that maximality implies $W=V$. For if $W=V$ did
  not hold, then we could find $a\in V$ with $a\not\in W$, consider
  $W'\defin W+\F a$, and extend $\rho$ to $\rho':W'\to\R$ by choosing
  \begin{equation}
    \label{HBinterval}
    \rho'(a) \in \left[ \sup_{b\in W \textrm{ s.t. } b + a \geq 0}
      (-\rho(b)), \inf_{b\in W \textrm{ s.t. } b - a \geq 0} \rho(b)
    \right]. 
  \end{equation}
  If the interval is non-empty, then we obtain $\rho':W'\to\R$ by
  linear extension. So it remains to show that the interval is
  non-empty, and that the resulting $\rho'$ is indeed a probability
  measure, since then we conclude $(W,\rho)<(W',\rho')$ in contradiction with the maximality assumption. We start with the first. The supremum is not $+\infty$
  since $u\in W$ and there is $\lambda\in\F$ with $\lambda u+a\geq
  0$; likewise, the infimum is not $-\infty$. It remains to be shown
  that its lower end is less than or equal to its upper end by showing
  that $-\rho(b_1)\leq \rho(b_2)$ for $b_1+a\geq 0$ and $b_2-a\geq
  0$. But the latter two conditions imply that $b_1+b_2\geq 0$, and
  hence $\rho(b_1) + \rho(b_2)\geq 0$ by the assumptions on $\rho$.

  Finally, we still need to show that $\rho'$ is indeed a probability
  measure, i.e.~that it is positive. But this is easy,
  since~\eqref{HBinterval} was engineered to guarantee precisely this:
  any element $c\in W'$ with $c\geq 0$ is a unique linear combination
  $c=\lambda a + b$ with $\lambda\in\F$ and $b\in W$. For $\lambda=0$,
  we have $\rho'(c)=\rho(c) \geq 0$, since $\rho$ is a probability
  measure. Otherwise, we can assume $\lambda=\pm 1$ after
  rescaling. If $\lambda=+1$, then we have $c = b + a \geq 0$, and the
  claim follows from $\rho'(a)\geq -\rho(b)$. If $\lambda=-1$, then we
  have $c = b - a$, and the claim follows from $\rho'(a)\leq \rho(b)$.
\end{proof}

Since there is a unique probability measure on the one-dimensional
subspace $U\defin \F u$, we obtain immediately:

\begin{cor}
  \label{HBsimple}
  There is a probability measure $\rho:V\to\R$.
\end{cor}

Our goal in the following is to derive a refined version of this
statement.

\begin{thm}[Hahn-Banach separation theorem]
  \label{HBseparate}
  Let $(V,C,u)$ be an order unit space and $a\in V$ an element for
  which there exists $n\in\N$ with $a+\tfrac{1}{n}u\not\geq 0$. Then
  there is a probability measure $\rho:V\to\R$ with $\rho(a) < 0$.
\end{thm}

\begin{proof}
  If $a$ is a scalar multiple of $u$, then it must be a negative
  scalar multiple, and the claim follows from
  Corollary~\ref{HBsimple}. Otherwise, the subspace $U\defin \F u + \F
  a$ is two-dimensional, and this is the case that we consider from
  now on.

  By Theorem~\ref{HBextend}, it is sufficient to define $\rho$ on $U$
  only. Obtaining such a $\rho$ can be done as in the proof of
  Theorem~\ref{HBextend}: there is a unique state on the
  one-dimensional subspace $\F u$, and we extend this state to $U$ by
  finding a feasible value for $\rho(a)$ and extending linearly. The
  interval of feasible values~\eqref{HBinterval} now becomes
  \[
  \rho(a) \in \left[ \sup_{\lambda\in\F \textrm{ s.t. } \lambda u + a
      \geq 0} (-\lambda), \inf_{\lambda\in\F \textrm{ s.t. } \lambda u
      - a \geq 0} \lambda \right],
  \]
  and we have already shown above that this is a non-empty
  interval. However, we also would like to achieve $\rho(a) < 0$; in
  order for this to be possible, we need the lower end of the interval
  to be strictly negative, or equivalently
  \[
  \inf_{\lambda\in\F \textrm{ s.t. } \lambda u + a \geq 0} \lambda >
  0.
  \]
  But this is equivalent to the assumption of existence of an $n\in\N$
  such that $\tfrac{1}{n} u + a \not\geq 0$.
\end{proof}

\section{Proof of main results}

\label{Proofs}

\mainloc*

\begin{proof}
  Our strategy is to apply \cref{HBsimple} to a suitably constructed
  order unit space.

  Let $V$ be the rational vector space with basis $X$, i.e.~the
  elements of $V$ are finite $\Q$-linear combinations of the
  outcomes. We write $e_x$ for the basis vector associated to an
  outcome $x\in X$. If $A$ is an event, then we also write $e_A$ as
  shorthand for $\sum_{x\in A} e_x$. In particular, we have
  $e_\emptyset = 0$. We define a convex cone $C$ in this space as the
  set of all finite \emph{non-negative} linear combinations of vectors
  of the form
  \begin{equation}
    \label{conedefn}
    e_A - e_B
  \end{equation}
  for all pairs of events $A$ and $B$ with $\Pl(A)\succeq\Pl(B)$. The
  idea here is that if we can find a vector $u$ such that $(V,C,u)$ is
  an order unit space, then, by \cref{HBsimple}, there exists a
  probability measure $\rho:V\to\R$.  This induces a probability
  measure $\mu$ on the test space via $\mu(A) = \rho(e_A)$.

  In more detail, this works as follows: we fix a test $T\in\Sigma$
  and consider the distinguished element $u\defin e_T$, claiming that
  this $u$ is an order unit. In order to show that it indeed is, we
  need to prove that any other vector $\sum_x \alpha_x e_x$ can be
  lower bounded by a scalar multiple of $u$. If two vectors can be
  lower bounded by a multiple of $u$, then so can any positive linear
  combination of these vectors; therefore, it is sufficient to
  consider the case of a single basis vector $e_x$ or $-e_x$. The
  first is trivial: with $A=\{x\}$ and $B=\emptyset$, the inequality
  $e_x\geq 0$ is itself an instance of~\eqref{conedefn}. For the
  second, choose a test $S\in\Sigma$ with $x\in S$. Then $-e_x \geq
  -e_S$ is again an instance of~\eqref{conedefn}, since $\Pl(S)\succeq
  \Pl(\{x\})$. Moreover, since all tests have plausibility $1$, we
  also have $-e_S\geq -e_T=-u$ as an instance of~\eqref{conedefn}. In
  conclusion, we obtain $-e_x\geq -u$.

  In order to complete the verification of the order unit space
  axioms, we still need to make sure that $-u\not\geq 0$, i.e.~that
  $-e_T$ is not a positive linear combination of vectors of the
  form~\eqref{conedefn}. This is where the assumption that $\Pl$ is
  Archimedean comes in. In fact, we will prove something more general:
  if $A$ and $B$ are events with $\Pl(A) \prec \Pl(B)$, then $e_A -
  e_B$ is not in our convex cone; the statement $-u\not\geq 0$ is then
  the special case with $A=\emptyset$ and $B=T$. So assume that $e_A -
  e_B$ is in our convex cone. This means that there are events
  $(A_1,\ldots,A_n)$ and $(B_1,\ldots,B_n)$ such that
  \begin{equation}
    \label{decomp}
    e_A - e_B = \sum_i \lambda_i (e_{A_i} - e_{B_i}),
  \end{equation}
  where $\lambda_i \in \Q$ and $\Pl(A_i) \succeq \Pl(B_i)$.  Writing
  each $\lambda_i$ in lowest terms as $\lambda_i = \alpha_i/\beta_i$
  and multiplying \cref{decomp} by $N$, where $N$ the least common
  multiple of the $\beta_i's$, gives
  \begin{equation}
    N e_A - N e_B = \sum_i r_i (e_{A_i} - e_{B_i}),
  \end{equation}
  where each $r_i$ is a positive integer.  Defining the list
  $(A'_1,\ldots,A'_m)$ so that the first $r_1$ events are $A_1$, the
  next $r_2$ events are $A_2$, etc., and similarly for
  $(B'_1,\ldots,B'_m)$, we obtain the decomposition
  \begin{equation}
    N e_A - N e_B = \sum_i (e_{A'_i} - e_{B'_i}),
  \end{equation}
  or equivalently
  \begin{equation}
    \label{equiv}
    \sum_i e_{B'_i} + N e_A = \sum_i e_{A'_i} + N e_B.
  \end{equation}
  If we construct the lists $(A'_1,\ldots,A'_m,B, \ldots, B)$ and
  $(B'_1,\ldots,B'_m,A,\ldots,A)$ by appending $N$ copies of $B$ and
  $A$ to the end of the lists $(A'_1,\ldots,A'_m)$ and
  $(B'_1,\ldots,B'_m)$ respectively, then \cref{equiv} says exactly
  that each $x \in X$ occurs the same number of times in these two
  lists.  By construction we have $\Pl(A'_i)\succeq\Pl(B'_i)$ and
  $\Pl(B) \succ \Pl(A)$. But applying the Archimedean property of
  \cref{archim} then gives $\Pl(B)\preceq \Pl(A)$, which is a
  contradiction.
\end{proof}

\mainfin*

\begin{proof}
  The only if part follows from the properties of probability
  measures, so we focus on the if part.

  For finite
  $X$, the convex cone constructed in the proof of \cref{mainloc} is
  finite-dimensional and polyhedral. Hence by Theorem~\ref{HBseparate}
  (or simply Farkas' lemma), any point outside of the cone can be
  separated from the cone by a probability measure. In particular,
  this applies to the vector $e_A - e_B$ for any pair of events with
  $\Pl(A) \prec \Pl(B)$, which we previously proved to be outside of
  the cone. So let $\rho_{A,B}$ be such a probability measure with
  $\rho_{A,B}(e_A - e_B) < 0$. We then consider the new probability
  measure
  \[
  \rho' \defin \frac{1}{N} \sum_{A,B \textrm{ s.t. } \Pl(A) \prec
    \Pl(B)} \rho_{A,B},
  \]
  where $N$ is the appropriate normalization factor equal to the
  number of terms in the sum. We claim that $\rho'(e_A - e_B) < 0$ for
  any $A$ and $B$ with $\Pl(A) \prec \Pl(B)$. For any two pairs of
  events $A,B$ and $A',B'$ with $\Pl(A) \prec \Pl(B)$ and $\Pl(A')
  \prec \Pl(B')$, we have $\rho_{A',B'}(e_B - e_A) \geq 0$ since
  $e_B-e_A \geq 0$ and $\rho_{A',B'}$ is a positive functional. By
  linearity, we then have $\rho_{A',B'}(e_A - e_B) = -
  \rho_{A',B'}(e_B - e_A) \leq 0$.  But since $\rho_{A,B}(e_A - e_B) <
  0$ by definition, we obtain the claim. Setting $\mu(A) = \rho(e_A)$
  results in a probability measure on the test space which satisfies
  \begin{align*}
    \Pl(A) \preceq \Pl(B) \quad \Longrightarrow & \quad \mu(A) \leq \mu(B) , \\[5pt]
    \Pl(A) \prec \Pl(B) \quad \Longrightarrow & \quad \mu(A) < \mu(B) .
  \end{align*}
  To see that $\mu$ also satisfies the converse implication, we make use of the assumption of total
  ordering.  We know that, for every pair of events $A$ and $B$,
  either $\Pl(A) \preceq \Pl(B)$ or $\Pl(B) \preceq \Pl(B)$, and that
  this is reflected in the constructed probability measure $\mu$.
  Therefore, whenever $\mu(A) \leq \mu(B)$, this necessitates $\Pl(A) \preceq \Pl(B)$.
\end{proof}

\section{Related work}

\label{Related}

The use of ordering relations rather than precise numerical
probabilities has a long history in the foundations of probability
theory, particularly in the intuitive probability of Koopman
\cite{Koopman1940} and in subjective Bayesian probability
\cite{Finetti1937, Ramsey1931}.  The notions of agreement and almost
agreement used in this work are derived from the literature on
axiomatizing subjective probability in these terms (see
\cite{Fishburn1986} for a review).  Most of this literature assumes a
total ordering, but see \cite{Keynes1921} for early arguments that a
partial ordering should be used instead.

The notion of a plausibility measure is due to Friedman and Halpern
\cite{Friedman1995, Friedman2001, Halpern2001}.  Their work was aimed
at providing a unifying framework for various mathematical structures
that had been proposed for reasoning in the face of uncertainty in
artificial intelligence, such as probabilities and Dempster-Shafer
belief functions \cite{Demster1967, Shafer1976}.  In doing so, they
effectively reinvented earlier theories of comparative probability
theory \cite{Fine1973}, except with partial rather than total
ordering, and in a far more elegant formalism.  We have followed their
notations here.

On the quantum side, Foulis, Randall and Piron investigated the notion
of \emph{supports} on test spaces \cite{Foulis1985}, which in our
terminology are just possibility measures on test spaces.  A
plausibility measure on a test space is a generalization of this, and
can be viewed as a way of unifying it with the usual notion of a
probability measure, or state, on a test space.

More recently, there has been much interest in applying possibilistic
measures in the foundations of quantum theory.  One of the authors of
the present paper has investigated possibilistic hidden variable
theories \cite{Fritz2009}.  This was followed by work of Abramsky
\cite{Abramsky2013}, investigating the logical structure of
non-locality and contextuality using possibility measures.  As both
possibility and probability are special cases of plausibility,
plausibility measures have the potential to unify Abramsky's approach
with the conventional account of these phenomena in terms of
probabilities.

Finally, Schumacher and Westmoreland have developed a version of
quantum theory on vector spaces over finite fields, which they call
\emph{modal} quantum theory \cite{Schumacher2010, Schumacher2010a,
  Schumacher2012}.  They get around the problem of having no inner
product on such spaces by basing their theory on possibilistic
assignments rather than probabilities.  They have noted that not all
of the possibility measures in their theory can be represented by
probabilities \cite{Schumacher2010a, Schumacher2012} (i.e.\ that they
don't agree with any probability measure in our terminology) and
raised the question of how such cases could be identified.  The
Archimedean condition can certainly be used for this as, in
particular, \cref{agreefail} occurs in their theory for a modal
quantum bit.  However, this does not yield an efficient algorithm for
checking agreement, as one potentially has to consider all possible
families of events.  Employing the order unit space construction
directly may be preferable.

\section{Conclusion}

\label{Conc}

In this paper, we have taken the first steps in developing the theory
of plausibility measures for general operational theories, which
include quantum theory as a special case.  We have shown that the
Archimedean property, which successfully identifies when classical
plausibility measures agree with classical probability measures, is
still useful in general theories.  In particular, for theories based
on strictly finite test spaces, such as Schumacher-Westmoreland modal
quantum theory, it provides a necessary and sufficient condition for
agreement.  For locally finite test spaces, we have shown that it is
sufficient for almost agreement, but it is possible that this might be
improved.

There are many potential applications of plausibility measures, some
of which will be developed in future work.  In particular, they can be
used to unify possibilistic and probabilistic approaches to phenomena
in quantum foundations, such as nonlocality and contextuality.
Additionally, in quantum information science, there are many
applications in which reasoning with qualitative comparisons rather
than precise numerical probabilities may be beneficial.  For example,
perhaps quantum cryptography protocols can be proved secure using only
a few comparative assessments, or perhaps algorithms for numerically
simulating quantum systems can be rendered more efficient by only
tracking qualitative information, such as whether or not the system is
close to being in an eigenstate of some set of observables.  The
latter would be analogous to the classical artificial intelligence
applications for which Friedman and Halpern originally invented
plausibility measures \cite{Halpern2001a}.

More speculatively, plausibility measures offer the possibility of
developing future theories of physics using a weaker predictive
structure than probability, whilst still allowing for precise
probabilities in an appropriate limit.  This might be necessary in
future theories of quantum gravity.

\bibliographystyle{eptcs}
\bibliography{plausibility}

\end{document}